\documentclass[11pt]{article}
\usepackage{amssymb}%
\usepackage{amsmath}%

\usepackage{graphicx,amsmath,amsfonts,amscd,amssymb,bm,cite,epsfig,epsf,url,color,algorithm,algorithmic, bbm}
\usepackage{fullpage} \usepackage[small,bf]{caption}
\newtheorem{theorem}{Theorem}[section]
\newtheorem{lemma}[theorem]{Lemma}

\newcommand{\R}{\mathbb{R}}

\newcommand{\goto}{\rightarrow}
\renewcommand{\P}{\operatorname{\mathbb{P}}}
\newcommand{\E}{\operatorname{\mathbb{E}}}

\newcommand{\vct}[1]{\bm{#1}}
\newcommand{\mtx}[1]{\bm{#1}}

\newcommand{\Span}{\operatorname{span}}

\newcommand{\rank}{\operatorname{rank}}

\newcommand{\sgn}{\operatorname{sgn}}

\newcommand{\diag}{\operatorname{diag}}
\newcommand{\tr}{\operatorname{Tr}}

\numberwithin{equation}{section}

\def \endprf{\hfill {\vrule height6pt width6pt depth0pt}\medskip}
\newenvironment{proof}{\noindent {\bf Proof} }{\endprf\par}

\newcommand{\Tp}{T^\perp}
\newcommand{\PT}{\mathcal{P}_T}
\newcommand{\PTp}{\mathcal{P}_{T^\perp}}

\newcommand{\Op}{\Omega^{\perp}}

\newcommand{\PTO}{\mathcal{P}_{T \cap \Omega}}

\newcommand{\cA}{\mathcal{A}}

\title{Sparse Signal Recovery from Quadratic Measurements via Convex Programming}
\author{Xiaodong Li\footnote{Department of Mathematics, Stanford University, Stanford, CA 94305
}~~and Vladislav Voroninski\footnote{  Department of Mathematics, University of California, Berkeley, CA 94709
 }
}
\begin{document}

\date{September 2012}
\maketitle

\begin{abstract}
In this paper we consider a system of quadratic equations $|\langle \vct{z_j}, \vct{x}\rangle|^2=b_j, ~j=1,...,m$, where $\vct
{x} \in \mathbb{R}^n$ is unknown while normal random vectors $\vct{z_j} \in \mathbb{R}^n$ and quadratic measurements $b_j \in \mathbb{R}$ are known. The 
system is assumed to be underdetermined, i.e., $m<n$. We prove that if there exists a sparse solution $\vct{x}$ i.e., at most $k$ components of $\vct{x}$ are non-zero, then by solving a convex 
optimization program, we can solve for $\vct{x}$ up to a multiplicative constant with high probability, provided that $k\leq O(\sqrt{m\over{\log 
n}})$. On the other hand, we prove that $k \leq O(\log n\sqrt{m})$ is necessary for a class of naive convex 
relaxations to be exact.
\end{abstract}


\section{Introduction}

\subsection{Introduction and the main results}
Convex optimization methods have recently been proven to be very successful in solving some classes of linear or quadratic algebraic equations. One classical 
example is compressed sensing (\cite{CRT06, Donoho06}), where a system of underdetermined linear 
equations 
can be solved exactly by using an $\ell_1$-convex relaxation, provided that the unknown vector is sparse. A 
typical result is as follows:
\paragraph{Compressed Sensing}
Suppose $\mtx{A} \in \mathbb{R}^{m \times n}$ has IID $\mathcal{N}(0,1)$ entries and $\vct{x_0} \in \mathbb{R}^n$ satisfies 
$\|\vct{x_0}\|_0=k$ (only $k$ components of $\vct{x}$ are not zeros). If we have linear measurements $\vct{b}=\mtx{A}\vct{x_0}$, then we can recover $
\vct{x}$ 
exactly with high probability by solving
\begin{equation}
\label{eq:L1min}
 \begin{array}{ll}
    \text{minimize}    & \quad \|\vct{x}\|_1 \\
    \text{subject to} & \quad  \vct{b}=\mtx{A}\vct{x}
\end{array}
\end{equation}
provided $k\leq O(m/\log(n/m))$.\\
\\
Another example is a recently proposed semidefinite programming framework for phase retrieval, called PhaseLift \cite
{PhaseLift2}, by which a signal can be exactly recovered-up to a multiplicative constant- from quadratic measurements. The SDP is a combination of 
trace minimization and Shor's SDP-relaxation for quadratic constraints. We review the results in \cite{PhaseLift2, ImprovedPhaseLift} 
below:

\paragraph{PhaseLift}
Fix a signal $\vct{x} \in \mathbb{R}^{n}$. Let $\vct{z_i} \in \mathbb{R}^n$ be IID standard normal random vectors, and 
suppose $b_j, ~j=1,...,m$ are defined as follows:
\begin{equation}
\label{basic system}
b_j=|\langle \vct{z_j}, \vct{x}\rangle|^2,~j=1,...,m,
\end{equation}
If we assume $m\geq C_0n$ for some numerical constant $C_0$, then with high probability, $\vct{x}\vct{x}^T$ is the unique 
solution to the following convex optimization problem:
\begin{equation}
\label{eq:tracemin}
 \begin{array}{ll}
    \text{minimize}    &\quad  \tr(\mtx{X}) \\
    \text{subject to} & \quad  \vct{z_j}^T\mtx{X}\vct{z_j} = b_j,~~j=1,..,m,\\
                                    & \quad \mtx{X} \succeq 0. 
\end{array}
\end{equation}
Notice that $\vct{x}\vct{x}^T$ is feasible since $\vct{x}\vct{x}^T\succeq \mtx{0}$ and
\[
\vct{z_j}^T(\vct{x}\vct{x}^T)\vct{z_j}=|\langle \vct{z_j}, \vct{x}\rangle|^2=b_j, ~j=1,...,m. 
\]

There is an inherent ambiguity to the solution of \eqref{basic system}, since multiplying by a phase factor ($\pm 1$ in the real case) does 
not change measurements. From now on, we only consider solutions modulo multiplication by phase. \\

In this paper , we consider model \eqref{basic system} in the case that $m<<n$. In this regime, \eqref{basic system} does not yield injective 
measurements. In fact, each equation in \eqref{basic system} is the union of two linear equations by assigning different signs, so generally we have $2^m$ solutions. However, if we assume that the unknown vector $\vct{x}$ is $k$-sparse, then under some mild conditions on the number of measurements, system \eqref{basic system} becomes well-posed:

\begin{theorem}
\label{well-posedness}
Let $\vct{x} \in \mathbb{R}^{n}$ be a k-sparse real signal, $\vct{a_i} \in \mathbb{R}^n , i = 1\ldots, m_1$ be generic real 
measurement vectors and let $\vct{y} \in \mathbb{C}^{n}$ be a $k$-sparse complex signal and $\vct{b_i} \in 
\mathbb{C}^n, i = 1\ldots m_2$ be generic complex measurement vectors. Then $m_1 \geq 4k-1$, $m_2 \geq 8k-2$ quadratic 
measurements $\{\left \langle\vct{a_i}, \vct{x}\right \rangle^2\}_{i=1}^{m_1}$, $\{|\left \langle\vct{b_i}, \vct{y}\right \rangle|^2\}_{i=1}^
{m_2}$ 
are sufficient to recover $\vct{x}$ and $\vct{y}$ modulo phase.
\end{theorem}
By generic we mean an open dense subset of the set of all $m$-element frames in $\mathbb{R}^n$ or $\mathbb{C}^n$.\\
\\
\begin{proof}
We only prove the complex case, since the real case is similar. Assume that there is a $k$-sparse $\vct{y'} \in \mathbb{C}^n$ 
such that $\left|\langle\vct{b_i}, \vct{y'} \rangle \right|^2 = \left| \langle \vct{b_i}, \vct{y} \rangle \right|^2, i = 1,\ldots m_2 \geq 
8k-2$. Let 
$T$ be the union of the supports of $\vct{y}$ and $\vct{y'}$. Clearly $|T| \leq 2k$. Then
\[
\left|\langle\vct{b_i}, \vct{y} \rangle \right|^2=\left| \langle \vct{b_i}, \vct{y'} \rangle \right|^2, ~~i=1,\ldots,m_2, 
\]
which is equivalent to 
\[
\left|\langle {\vct{b_j}}_T, \vct{y}_T \rangle \right|^2=\left| \langle {\vct{b_j}}_T, \vct{y'}_T \rangle \right|^2, ~i=1,...,m_2,
\]
where $\vct{v}_T$ means the restriction of $\vct{v}$ on the support $T$. The genericity of $\vct{b_i}, ~i=1,...,m_2$ implies the 
genericity of ${\vct{b_i}}_T, ~i=1,...,m$. Then since $m_2 \geq4(2k)-2 = 8k-2$ we have $\vct{y}_T=e^{i\psi}\vct{y'}
_T$ 
for some real number $\psi$ by Theorem 3.1 in \cite{Edidin}. Therefore $\vct{y}=e^{i\psi}\vct{y'}$. 
\end{proof}

Injectivity of the measurements of course doesn't imply that efficient recovery is possible. Yet, inspired by the success of convex 
relaxations in compressed sensing and phase retrieval, it is natural to leverage the sparsity assumption to try to efficiently recover signals from fewer than $n$ intensity measurements. A convex formulation in this direction, which, to the best of our knowledge, was first proposed in \cite{CPR} to solve \eqref{basic system}, is the following program:
\begin{equation}
\label{eq:traceL1min}
 \begin{array}{ll}
    \text{minimize}    & \quad \|\mtx{X}\|_1 + \lambda\tr(\mtx{X}) \\
    \text{subject to} & \quad  \vct{z_j}^T\mtx{X}\vct{z_j} = b_j,~~j=1,..,m,\\
                                    & \quad \mtx{X} \succeq 0. 
\end{array}
\end{equation}
The next theorem shows 
that when $\vct{z_j}$ are IID standard normal random vectors,  the solution to \eqref{eq:traceL1min} for 
an 
appropriate choice of $\lambda$, is exactly $\vct{x}\vct{x}^T$, provided that $k\leq O(\sqrt{m\over{\log n}})$.

\begin{theorem}
\label{thm1}
Fix a signal $\vct{x} \in \mathbb{R}^{n}$ with $\|\vct{x}\|_2 =1$ and $\|\vct{x}\|_0=k$, i.e, only $k$ components of $\vct{x}$ 
are non-zero. Let $\vct{z_i} \in \mathbb{R}^n$ be IID standard normal random vectors, and suppose $b_j, 
~j=1,...,m
$ are defined as in \eqref{basic system}.  Then the solution to the convex program \eqref{eq:traceL1min} is
exact with probability at least $1 - ( 2\log n+3 )(4e^{-\gamma \frac{m}{ 2\log(n) +3}}+{1\over{n^3}})-(5+2n^2)e^{-\gamma m}$, provided
$\lambda >\sqrt{k}\|\vct{x}\|_1+1$, $\lambda < {{n^2}\over{4}}$ and $m>C_0\lambda^2\log n$. Here $C_0$ and $
\gamma$ are numerical constants.
\end{theorem}
Remark 1: By choosing $\lambda=\sqrt{m\over{4C_0\log n}}$, we have exact recovery with probability at least $1 - (2\log n+3)(4e^{-\gamma \frac{m}{ 2\log(n) +3 }}+{1\over{n^3}})-(5+2n^2)e^{-\gamma m}$ if the number of measurements obeys 
$m 
\geq O(\|\vct{x}\|_1^2k\log n)$. Moreover, by choosing $\vct{x}$ to be a k-sparse vector with components $x_i=\pm
{1\over{\sqrt{k}}}$, this reads $m\geq O(k^2\log n)$. \\
\\
Remark 2: In \cite{CPR}, the authors operate under an assumption that the sampling operator satisfies a generalization of the Restricted 
Isometric Property and mutual coherence, while in Theorem \ref{thm1} of our paper we assume the $\vct{z_j}$'s 
are 
IID standard Normal vectors. In our setting the mutual coherence of the sampling operator defined in \cite{CPR} will be on the 
order of $O(1)$, since the diagonal entries of $\vct{z_j}\vct{z_j}^T$ are always $\chi^2$ random variables. 
Applying 
the result in \cite{CPR} we get $k=O(1)$ in our setting, which is a much smaller range of sparsity than considered in the result of the above theorem.\\
\\
The conclusion of Theorem \ref{thm1} is far more restrictive than that of Theorem \ref{well-posedness}, so one may ponder whether \ref{thm1} is optimal. The following result shows 
that indeed there is a substantial gap between solving  \eqref{basic system} and \eqref{eq:traceL1min}. 

\begin{theorem}
\label{thm2}
Under the setting of Theorem \ref{thm1}, assuming $4 \leq k\leq m \leq {n\over{40\log n}}$, then there is an event $E$ with probability at least $1-{m\over{n^5}}-m e^{-0.09n+0.09k+0.79m}$, such that the following property holds: If there exists a $\lambda \in \mathbb{R}$ such that $\vct{x}\vct{x}^T$ is a minimizer of \eqref{eq:traceL1min}, then we have 
\[
m\geq  \min\left(({k\over 4}-1)^2, {{\max(\|\vct{x}\|_1^2-k/2, 0)^2}\over{500\log^2 n}}\right).
\]
\end{theorem}
Remark: Taking $\vct{x}$ to be a k-sparse vector with components $x_i=\pm{1\over{\sqrt{k}}}$, this reads $m\geq O
(k^2/\log^2 n)$.\\
\\
This theorem obtains sharp theoretical results on the performance of \eqref{eq:traceL1min} in the Gaussian quadratic measurement setting, which may be surprising since it implies that there is a substantial gap between the sufficient number of measurements for injectivity and the necessary number of measurements for recovery via a class of natural convex relaxations. \\

\subsection{Definitions and notations}
In this section we introduce some useful definitions and notations, which will be used in the proofs of Theorems \ref{thm1} and \ref{thm2}. In this paper vectors and matrices are boldfaced while scalars are not. \\
\\
For any positive integer $n_0$, denote $[n_0] = \{1, \ldots, n_0\}$. Let $G=\{i \in [n]: x_i\neq 0\}$ be the support of $\vct{x}$ and $B$ be the complement $G=\{i \in [n]: x_i= 0\}$. Without loss of generality, we assume $G=\{1,...,k\}$. Define the subspaces of symmetric matrices $\{X_{ij}=0, i>k \text{~or~}j>k, \mtx{X}=\mtx{X}^T\}$, $\Gamma=\{\mtx{X}|X_{ij}=0, i\leq k \text{~or~}j\leq k, \mtx{X}=\mtx{X}^T\}$ and $T=\{\vct{x}\vct{x_0}^T+\vct{x_0}\vct{x}^T, \vct{x_0} \in \mathbb{R}^n\}$. In the space of symmetric matrices, we define the inner product $\langle \mtx{X}, \mtx{Y}\rangle=\tr(\mtx{X}\mtx{Y})$. Then for any subspace of symmetric matrices $R$, we denote by $R^\perp$ its orthogonal complement under such an inner product.\\
\\
For the given random vectors $\vct{z_j}, ~j=1,...,m$, let
$\cA: \R^{n \times n} \goto \R^m$ be the linear operator $\cA(\mtx{X}) =
\{\tr(\vct{z}_i \vct{z}_i^T \mtx{X})\}_{i \in [m]}$ for any symmetric matrix $\mtx(X)$. Hence its adjoint is $\cA^*(\vct{y}) = \sum_{i \in [m]} y_i \vct{z}_i
\vct{z}_i^T$. \\
\\
For a symmetric matrix $\mtx{X}$, we put
$\mtx{X}_T$ for the orthogonal projection of $\mtx{X}$ onto $T$ and
similar to $\mtx{X}_{T^\perp}$, $\mtx{X}_{\Omega}$, $\mtx{X}_{\Omega^\perp}$, $\mtx{X}_{\Gamma^\perp}$, $\mtx{X}_
{\Omega \cap T}$ and so on.
For a vector $\vct{v} \in \mathbb{R}^n$, we define $\vct{v}_G=\langle \vct{v}, \vct{e_1} \rangle \vct{e_1}+...+\langle \vct
{v}, \vct{e_k} \rangle \vct{e_k}$ and $\vct{v}_B=\vct{v}-\vct{v}_G$. Here $(\vct{e_1}, ..., \vct{e_n})$ is the 
standard 
basis of $\mathbb{R}^n$.\\
\\
Denote $\|\vct{y}\|_p$ as the $\ell_p$ norm of a
vector $\vct{y}$, where $p$ could be $0$, $1$ or $2$. 
Let $\|\mtx{X}\|$ and  $\|\mtx{X}\|_F$ be the
spectral and Frobenius norms of a matrix $\mtx{X}$, respectively. Moreover, let $\|\mtx{X}\|_\infty$ and $\|\mtx{X}\|_1$ be the 
maximum and the summation of absolute values of all entries of $\mtx{X}$ respectively, i.e., they represent the $\ell_\infty$ and $\ell_1$ norms of the vectorizations of matrices.

\section{The proof of Theorem \ref{thm1}}
In this section we will prove Theorem \ref{thm1}. First we will cite and prove some supporting lemmas. Then we prove that it 
suffices to construct an approximate dual certificate matrix to the primal convex optimization problem. Finally we use a 
modification of the golfing scheme to construct such an approximate dual certificate with high probability. Both the idea of the 
approximate dual certificate and the golfing scheme are originally due to David Gross' work \cite{Gross09} in Matrix 
completion.

\subsection{Preliminaries} 
In this section we establish some useful properties of $\cA$.
\begin{lemma}[\cite{PhaseLift2}]
\label{lem:PL}
There is an event $E$ of probability at least $1 - 5e^{-\gamma_0 m}$
such that on $E$, any positive symmetric matrix obeys 
\begin{equation}
\label{ineq:l1-trace}
(1 - 1/8) \tr(\mtx{X_\Omega}) \leq m^{-1} \|\cA(\mtx{X_\Omega})\|_1 \le (1 + 1/8) \tr(\mtx{X_\Omega}),  
\end{equation}
and any symmetric rank-2 matrix obeys
\begin{equation}
\label{ineq: l1-lowrank-mlarge}
m^{-1} \|\cA(\mtx{X_\Omega})\|_1 \ge 0.94 (1 - 1/8) \|\mtx{X_\Omega}\|.
\end{equation}
\end{lemma}

\begin{lemma}
\label{lem:l1}
There is an event $E$ of probability at least $1 - 2n^2e^{-\gamma_0 m}$
such that on $E$, any symmetric matrix obeys 
\begin{equation}
\label{ineq: l1-lowrank-msmall}
m^{-1} \|\cA(\mtx{X})\|_1\leq {9\over 8}\|\mtx{X}\|_1.
\end{equation}
\end{lemma}
\begin{proof}
By direct calculation,  we have
 \begin{align*}
 \frac{1}{m} \| \cA(\mtx{X}) \|_1 &= \frac{1}{m}  \sum_{j=1}^m \left| \left<\mtx{X},\vct{ z_j}\vct{ z_j}^T\right>\right|\leq 
\frac{1}{m}\sum_{j=1}^m \sum_{a,b} |X_{ab}z_{ja}{z_{jb}}|\leq \max_{a,b}{1\over m}(\sum_{j=1}^m|z_{ja}{z_{jb}}|)
\|
\mtx{X}\|_1.
 \end{align*}
 Since $|z_{ja}{z_{jb}}|, ~ j=1...,m$ are IID sub-exponential variables with expectation $1$ or ${2\over \pi}$ and have finite $\psi_1$-norm. By Proposition 5.16 of \cite{VershyninRMT}, we have
 \[
 \max_{a,b}{1\over m}(\sum_{j=1}^m|z_{ja}{z_{jb}}|)\leq {9/8}
 \]
 with probability at least $1 - 2n^2e^{-\gamma_0 m}$. On this event we have $m^{-1} \|\cA(\mtx{X})\|_1\leq {9\over 8}\|
\mtx{X}\|_1$.
\end{proof}

\subsection{Exact recovery by the existence of an approximate dual certificate.}
In the classical theory of semidefinite programming, the existence of an exact dual certificate can be used to prove that a 
specific point is the solution to the primal problem. By using an idea in \cite{Gross09}, in order to prove Theorem 
\ref
{thm1}, it suffices to prove the existence of an approximate dual certificate.

\begin{lemma}
\label{exact recovery}
Denote $\mtx{X_0}=\lambda \vct{x}\vct{x}^T+\PT(\sgn(\vct{x})\sgn(\vct{x})^T)$. Suppose there exists $\mtx{Y}=v_1\vct
{z_1}\vct{z_1}^T+...+v_m\vct{z_m}\vct{z_m}^T$ for some real numbers $v_1, ..., v_m$ satisfying
$\|\mtx{Y}_{T\cap \Omega}-X_0\|_F\leq{{\|\mtx{X_0}\|_F}\over{6n^2}}$, $\|\mtx{Y}_{T^\perp \cap \Omega}\|\leq {{\|
\mtx{X_0}\|_F}\over 5}$ and $\|\mtx{Y}_{\Omega^{\perp}}\|_\infty\leq {{C\sqrt{\log n}}\over{\sqrt{m}}}\|
\mtx
{X_0}\|_F$,
with some numerical constant $C$. Then assuming that $\cA$ satisfies properties \eqref {ineq:l1-trace}, \eqref{ineq: l1-lowrank-mlarge} and \eqref{ineq: l1-lowrank-msmall}, we have that $\vct{x}\vct{x}^T$ is the unique solution to the 
convex 
program (\ref{eq:traceL1min}), provided that $\lambda >\sqrt{k}\|\vct{x}\|_1+1$, $\lambda < {{n^2}\over{4}}$ and 
$m>64C^2\lambda^2\log n$.
\end{lemma}

\begin{proof}
 Let $\mtx{\hat{X}}$ be the solution to the convex program (\ref{eq:traceL1min}) and let $\mtx{H}=\mtx{\hat{X}}-\vct{x}\vct
{x}^T$. Then by the feasibility condition of the convex program (\ref{eq:traceL1min}) , we have
 \begin{equation}
 \label{equality constraint}
 \cA(\mtx{H})=0,
\end{equation} 
and
\begin{equation}
\label{inequality constraint}
 \vct{x}\vct{x}^T+ \mtx{H} \succeq 0. 
 \end{equation}
 By inequality (\ref{inequality constraint}), we have 
\begin{equation}
\label{ineq: positive}
\mtx{H}_{T^\perp\cap \Omega}\succeq 0, ~~~\mtx{H}_{B}\succeq 0 ~~\text{and}~~\mtx{H}_{T^\perp}\succeq 0.
\end{equation}  
By equality (\ref{equality constraint}), we have $\cA(\mtx{H}_{T\cap \Omega})=\cA(\mtx{H}_{T^\perp \cup \Omega^
\perp})$. Then by (\ref{ineq:l1-trace}), (\ref{ineq: l1-lowrank-mlarge}), (\ref{ineq: l1-lowrank-msmall}) and (\ref
{ineq: 
positive}), we have
\begin{align*}
\|\mtx{H}_{T\cap \Omega}\| &\leq {1\over{0.94\times (7/8)}}\frac{1}{m}\|\cA(\mtx{H}_{T\cap \Omega})\|_1 \notag\\
                                                           & \leq \frac{1.3}{m}\|\cA(\mtx{H}_{T^\perp \cup \Omega^\perp})\|_1 \notag\\
                                                           & \leq \frac{1.3}{m}(\|\cA(\mtx{H}_{T^\perp\cap \Omega})\|_1+\|\cA(\mtx{H}_{\Omega^
\perp})\|_1) \notag\\
                                                           &\leq  {1.3\times(9/8)}\left( \tr(\mtx{H}_{T^\perp\cap \Omega}) + \|\mtx{H}_{\Omega^
\perp}\|_1 \right).
\end{align*}
Since $\rank(\mtx{H}_{T\cap \Omega})\leq 2$, we have
\begin{equation}
\label{ineq: upperbound}
\|\mtx{H}_{T\cap \Omega}\|_F\leq \sqrt{2}\|\mtx{H}_{T\cap \Omega}\| \leq 2.5\left( \tr(\mtx{H}_{T^\perp\cap 
\Omega}) + \|\mtx{H}_{\Omega^\perp}\|_1 \right).
\end{equation}

Now let's see what inequalities about $\mtx{H}$ we can get from the objective function. Since both $\mtx{\hat{X}}$ and $\vct
{x}\vct{x}^T$ are feasible and $\mtx{\hat{X}}$ is the minimizer, we have
\[
\|\mtx{\hat{X}}\|_1+\lambda \tr(\mtx{\hat{X}})\leq\|\vct{x}\vct{x}^T\|_1+\lambda \tr(\vct{x}\vct{x}^T).
\]
Also, since
\begin{align*}
\|\mtx{\hat{X}}\|_1+\lambda \tr(\mtx{\hat{X}})     & =      \|\vct{x}\vct{x}^T+\mtx{H}\|_1+\lambda \tr(\vct{x}\vct{x}^T+
\mtx{H})\\
                                                                                                  &\geq \|\vct{x}\vct{x}^T\|_1+\langle \sgn(\vct{x})\sgn(\vct{x})^T, \mtx
{H}\rangle+\|\mtx{H}_{\Omega^\perp}\|_1+\lambda\tr(\vct{x}\vct{x}^T)+\lambda\tr(\mtx{H}),
\end{align*}
we have
\begin{equation*}
\langle \sgn(\vct{x})\sgn(\vct{x})^T, \mtx{H}\rangle+\|\mtx{H}_{\Omega^\perp}\|_1+\lambda\tr(\mtx{H})\leq 0.
\end{equation*}
This implies
\begin{equation*}
\langle \PT(\sgn(\vct{x})\sgn(\vct{x})^T)+\lambda\vct{x}\vct{x}^T, \mtx{H}_T\rangle+\langle \PTp(\sgn(\vct{x})\sgn(\vct
{x})^T), \mtx{H}_{T^\perp}\rangle+\|\mtx{H}_{\Omega^\perp}\|_1+\lambda\tr(\mtx{H}_{T^\perp})\leq 0.
\end{equation*}
It is easy to see that $\PTp(\sgn(\vct{x})\sgn(\vct{x})^T)$ is positive semidefinite and combining with (\ref{ineq: positive}), we 
get 
\[
\langle \PTp(\sgn(\vct{x})\sgn(\vct{x})^T), \mtx{H}_{T^\perp}\rangle\geq 0,
\]
which implies
\begin{equation*}
\langle \mtx{X_0}, \mtx{H}_{T\cap \Omega}\rangle+\|\mtx{H}_{\Omega^\perp}\|_1+\lambda\tr(\mtx{H}_{T^\perp})\leq 
0.
\end{equation*}
Notice that $\tr(\mtx{H}_{T^\perp})=\tr(\mtx{H}_{T^\perp\cap \Omega})+\tr(\mtx{H}_B)$. By \eqref{ineq: positive} and $
\lambda\geq 0$, we have
\begin{equation}
\label{ineq: optimality}
\langle \mtx{X_0}, \mtx{H}_{T\cap \Omega}\rangle+\|\mtx{H}_{\Omega^\perp}\|_1+\lambda\tr(\mtx{H}_{T^\perp\cap 
\Omega})\leq 0.
\end{equation}
By the construction of the approximate dual certificate $\mtx{Y}$, we know $\mtx{Y}=\cA^*(\vct{v})$, which implies $\langle 
\mtx{H}, \mtx{Y}\rangle=\langle \cA(\mtx{H}), \vct{v}\rangle=0$. Then we have
\[
\langle \mtx{H}_{T\cap \Omega}, \mtx{Y}_{T\cap \Omega}-\mtx{X_0} \rangle+\langle \mtx{H}_{T\cap \Omega}, \mtx{X}
_0\rangle+\langle \mtx{H}_{T^\perp\cap \Omega}, \mtx{Y}_{T^\perp\cap \Omega}\rangle+\langle \mtx{H}
_
{\Omega^\perp}, \mtx{Y}_{\Omega^\perp}\rangle=0.
\]
By the assumed properties of $\mtx{Y}$, we have
\[
{{\|\mtx{X_0}\|_F}\over{6n^2}}\|\mtx{H}_{T\cap \Omega}\|_F+\langle \mtx{H}_{T\cap \Omega}, \mtx{X}_0\rangle+{{\|
\mtx{X_0}\|_F}\over 5}\tr(\mtx{H}_{T^\perp\cap \Omega})+{{C\sqrt{\log n}}\over{\sqrt{m}}}\|\mtx{X_0}\|_F
\|
\mtx{H}_{\Omega^\perp}\|_1\geq 0.
\]
By \eqref{ineq: optimality}, we have
\begin{equation}
\label{prelowerbound}
{{\|\mtx{X_0}\|_F}\over{6n^2}}\|\mtx{H}_{T\cap \Omega}\|_F\geq (\lambda-{{\|\mtx{X_0}\|_F}\over 5})\tr(\mtx{H}_{T^
\perp\cap \Omega})+(1-{{C\sqrt{\log n}}\over{\sqrt{m}}}\|\mtx{X_0}\|_F)\|\mtx{H}_{\Omega^\perp}\|_1.
\end{equation}
Since 
\[
\PT(\sgn(\vct{x})\sgn(\vct{x})^T)=\|\vct{x}\|_1(\vct{x}\sgn(\vct{x})^T+\sgn(x)\vct{x}^T)-\|\vct{x}\|_1^2\vct{x}\vct{x}^T,
\]
we have
\[
\|\mtx{X_0}\|_F=\|\lambda \vct{x}\vct{x}^T+\PT(\sgn(\vct{x})\sgn(\vct{x})^T)\|_F\leq \lambda+\|\vct{x}\|_1^2+2\sqrt{k}
\|\vct{x}\|_1.
\]
Then together with the assumptions of $\lambda >\sqrt{k}\|\vct{x}\|_1+1$, $\lambda < {{n^2}\over{4}}$ and $m>64C^2\lambda^2\log n$, we have
\[
{{\|\mtx{X_0}\|_F}\over{6n^2}}\leq 3(\lambda-{{\|\mtx{X_0}\|_F}\over 5}) \text{~and~} {{\|\mtx{X_0}\|_F}\over{6n^2}}\leq 
3(1-{{C\sqrt{\log n}}\over{\sqrt{m}}}\|\mtx{X_0}\|_F),
\]
by direct calculation. Therefore, by \eqref{prelowerbound}
\begin{equation}
\label{ineq: lowerbound}
\|\mtx{H}_{T\cap \Omega}\|_F\geq 3\left( \tr(\mtx{H}_{T^\perp\cap \Omega}) + \|\mtx{H}_{\Omega^\perp}\|_1 
\right).
\end{equation}
 Equations \eqref{ineq: upperbound} and \eqref{ineq: lowerbound} give  $\mtx{H}_{T\cap \Omega}=0$, and then by \eqref
{ineq: lowerbound}, we have $\mtx{H}_{T^\perp\cap \Omega}=0$ and $\mtx{H}_{\Omega^\perp}=0$. 
Hence $
\mtx{H}=0$, which implies $\vct{x}\vct{x}^T$ is the unique minimizer of the convex program (\ref{eq:traceL1min}).
\end{proof}

\subsection{Key lemma}
The following lemma will be essential for the construction of a desirable dual certificate:
\begin{lemma}
\label{key lemma}
For any fixed $\mtx{X} \in T\cap \Omega$, we have $\rank(\mtx{X})\leq 2$. Consider an eigenvalue decomposition $\mtx{X}=
\lambda_1\vct{u_1}\vct{u_1}^T+\lambda_2\vct{u_2}\vct{u_2}^T$, where $\|\vct{u_1}\|=\|\vct{u_2}\|=1$, $
\vct
{u_1}^T\vct{u_2}=0$ and both $\vct{u_1}$ and $\vct{u_2}$ are supported on $G$. Define 
\begin{align*}
\mtx{Y}&=f(\lambda_1, \lambda_2, \vct{u_1}, \vct{u_2})\\
                & :={1\over m(\beta_4-\beta_2)}\sum_{j=1}^m (\lambda_1(|{\vct{z_j}}_G^T\vct{u_1}|^21_{\{|{\vct{z_j}}_G^T\vct
{u_1}|\leq 3\}}-\beta_2)+\lambda_2(|{\vct{z_j}}_G^T\vct{u_2}|^21_{\{|{\vct{z_j}}_G^T\vct{u_2}|\leq 3\}}-
\beta_2))
\vct{z_j}\vct{z_j}^T.
\end{align*}
Here we define $\beta_2=\E z^21_{\{|z|\leq 3\}}\approx 0.9707$, $\beta_4=\E z^41_{\{|z|\leq 3\}}\approx 2.6728$, where 
assuming $z$ a standard normal variable. Then with probability at least $1-4e^{-\gamma m}-1/{n^3}$, 
\[
\|\mtx{Y}_{T\cap \Omega}-\mtx{X}\|_F\leq {1\over 5}\|\mtx{X}\|_F,~~\|\mtx{Y}_{T^\perp\cap \Omega}\|\leq {1\over 
10}\|\mtx{X}\|_F\text{~and~} \|\mtx{Y}_{\Omega^{\perp}}\|_\infty\leq {{C_0\sqrt{\log n}}\over{\sqrt{m}}}
\|
\mtx{X}\|_F.
\]
provided $m\geq C_1k$. Here $\gamma$, $C_0$ and $C_1$ are numerical constants. 
\end{lemma}

Before proving Lemma \ref{key lemma}, we need to prove the following supporting lemma:
\begin{lemma}
\label{supporting}
Suppose $\vct{z_j} \in \mathbb{R}^n$, $j=1,...,m$ are IID $\mathcal{N}(0, \mtx{I}_{n \times n})$ random vectors, and $\vct{u}$ 
is any fixed vector with unit $2$-norm, i.e, $\|\vct{u}\|_2=1$. Then for any fixed $\epsilon>0$, there exists a 
constant 
$\gamma(\epsilon)$ and $C_0(\epsilon)$ satisfying 
\[
\left\|{1\over m}\sum_{j=1}^m (|{\vct{z_j}}^T\vct{u}|^21_{\{|{\vct{z_j}}^T\vct{u}|\leq 3\}}){\vct{z_j}}{\vct{z_j}}^T-((\beta_4-
\beta_2)\vct{u}\vct{u}^T+\beta_2\mtx{I})\right\|\leq \epsilon
\]
with probability at least $1-2e^{-\gamma m}$ provided $m\geq C_0n$.
\end{lemma}

\begin{proof}
By rotational invariance, we can assume $\vct{u}=\vct{e_1}$. Define a matrix $\mtx{D}=\diag({1\over{\sqrt{\beta_4}}}, 
{1\over{\sqrt{\beta_2}}}, ..., {1\over{\sqrt{\beta_2}}})$.
 Define $\vct{w_j}=\mtx{D}|z_{j1}1_{\{|z_{j1}|\leq 3\}}|\vct{z_j}$. It is immediate to check that the $\vct{w_j}$'s are IID copies 
of a
zero-mean, isotropic and sub-Gaussian random vector $\vct{w}$. Standard results about random matrices with sub-gaussian
rows---e.g.~Theorem 5.39 in \cite{VershyninRMT}---give 
\[
\left\|{1\over m}\sum_{j=1}^m {\vct{w_j}}{\vct{w_j}}^T-\mtx{I}\right\|\leq \epsilon/3,
\]
which implies
\begin{align*}
&\left\|\left({1\over m}\sum_{j=1}^m (|z_{j1}|^21_{\{|{z_{j1}}|\leq 3\}}){\vct{z_j}}{\vct{z_j}}^T-((\beta_4-\beta_2)\vct
{e_1}\vct{e_1}^T+\beta_2\mtx{I})\right)\right\|\\
&=\left\|\mtx{D}^{-1}\left({1\over m}\sum_{j=1}^m ({\vct{w_j}}{\vct{w_j}}^T-\mtx{I})\right)\mtx{D}^{-1}\right\|\leq \|\mtx
{D}^{-1}\|(\epsilon/3)\|\mtx{D}^{-1}\|\leq \epsilon.
\end{align*}
with probability at least $1 - 2 e^{-\gamma(\epsilon) m}$ provided that $m \ge C_0(\epsilon) n$, where $C_0$ is sufficiently 
large. 
\end{proof}

\begin{proof} of Lemma \ref{key lemma}.
It suffices to prove
\[
\|\mtx{Y}_\Omega-\mtx{X}\|\leq {\sqrt{2}\over{20}}\|\mtx{X}\|_F,~~~~ \|\mtx{Y}_{\Omega^{\perp}}\|_\infty\leq 
{{C_0\sqrt{\log n}}\over{\sqrt{m}}}\|\mtx{X}\|_F.
\]
since
\[
\|\mtx{Y}_{T\cap \Omega}-\mtx{X}\|_F\leq \sqrt{2}\|\mtx{Y}_{T\cap \Omega}-\mtx{X}\|\leq 2\sqrt{2}\|\mtx{Y}_
{\Omega}-\mtx{X}\|\leq {1\over 5}\|\mtx{X}\|_F,
\]
and
\[
\|\mtx{Y}_{T^\perp\cap \Omega}\|\leq {1\over 10}\|\mtx{X}\|_F.
\]
\paragraph{1. $\|\mtx{Y}_\Omega-\mtx{X}\|\leq {\sqrt{2}\over{20}}\|\mtx{X}\|_F$.} By Lemma \ref{supporting}, we have 
\[
\|{1\over m}\sum_{j=1}^m (|{\vct{z_j}}_G^T\vct{u_a}|^21_{\{|{\vct{z_j}}_G^T\vct{u_a}|\leq 3\}}){\vct{z_j}}_G{\vct{z_j}}_G^T-
((\beta_4-\beta_2)\vct{u_a}\vct{u_a}^T+\beta_2\mtx{I})\|\leq \epsilon, ~~a=1,2.
\]
with probability at least $1-2e^{-\gamma m}$ provided $m\geq C_1 n$. Similarly, since ${1\over m}\sum_{j=1}^m {\vct{z_j}}_G
{\vct{z_j}}_G^T$ is Wishart when restricted on $\Omega$, standard results in random matrix theory---
e.g.~Corollary 5.35 in
\cite{VershyninRMT}---assert that
\[
\|{1\over m}\sum_{j=1}^m {\vct{z_j}}_G{\vct{z_j}}_G^T- \mtx{I} \| \leq \epsilon
\]
with probability at least $1-2e^{-\gamma m}$ provided $m\geq C_1 n$. Then Denote 
\[
\mtx{W_a}={1\over m(\beta_4-\beta_2)}\sum_{j=1}^m (|{\vct{z_j}}_G^T\vct{u_a}|^21_{\{|{\vct{z_j}}_G^T\vct{u_a}|\leq 
3\}}-\beta_2){\vct{z_j}}_G{\vct{z_j}}_G^T-\vct{u_a}\vct{u_a}^T, ~~a=1,2.
\]
We have with probability at least $1-4e^{\gamma m}$, $\|\mtx{W_a}\|\leq {1\over{20}}$ provided $m\geq C_1k$. This 
actually gives us the conclusion by noticing that 
\[
\mtx{Y}_\Omega-\mtx{X}=\lambda_1\mtx{W_1}+\lambda_2\mtx{W_2}.
\]

\paragraph{2. $\|\mtx{Y}_{\Omega^{\perp}}\|_\infty\leq {{C_0\sqrt{\log n}}\over{\sqrt{m}}}\|\mtx{X}\|_F$.}~\\
For any fixed $a, b \in [n]$, $a>k$ or $b>k$, we know $Y_{ab}=\vct{e_a}^T\mtx{Y}\vct{e_b}$ is the arithmetic mean of m IID 
centered sub-exponential random variables, whose $\psi_1-$ norm is bounded by $K(|\lambda_1|+|
\lambda_2|)$ 
with a numerical constant $K$. Then by Proposition 5.16 in \cite{VershyninRMT}, we have 
\[
|Y_{ab}|_\infty\leq {{C_0\sqrt{\log n}}\over{\sqrt{m}}}\|\mtx{X}\|_F,
\]
with probability at least $1-1/{n^5}$, which implies our claim.
\end{proof}

\subsection{Adaptation of the golfing scheme}
In this section we will construct the dual certificate satisfying all the properties in Lemma \ref{exact recovery} by using the golfing 
scheme. \\

\begin{proof} of Theorem \ref{thm1}: It suffices to construct $\mtx{Y}$ satisfying all the properties in Lemma \ref{exact 
recovery} with high probability. We divide the group of IID random vectors $\{\vct{z_1}, ..., \vct{z_m}\}$ into $l:=
\lfloor 
2\log(n)\rfloor+3$ groups
\[
\{\vct{z^{(1)}_1}, ..., \vct{z^{(1)}_{m_1}}\}, ..., \{\vct{z^{(l)}_1}, ..., \vct{z^{(l)}_{m_l}}\}.
\]
This implies that $m_1+...+m_l=m$. We use the same definition of $\mtx{X_0}$ in Lemma \eqref{exact recovery}. For i=1,..,l, 
as in Lemma \ref{key lemma}, we define the eigenvalue decomposition
\[
\mtx{X_{i-1}}=\lambda_{1_{i-1}}\vct{u_{1_{i-1}}}\vct{u_{1_{i-1}}}^T+\lambda_{2_{i-1}}\vct{u_{2_{i-1}}}\vct{u_{2_{i-1}}}^T.
\]
 and 
 \[
 \mtx{Y_i}=f\left(\lambda_{1_{i-1}}, \lambda_{2_{i-1}}, \vct{u_{1_{i-1}}}, \vct{u_{2_{i-1}}}\right). 
\]
 Moreover, we define $\mtx{X_i}=\mtx{X_{i-1}}-\PTO({\mtx{Y_i}})$,  and $\mtx{Y}=\sum_{i=1}^l \mtx{Y_i}$. By 
definition we have $\mtx{X_i}$'s are in $T\cap \Omega$, so $\mtx{Y_i}$ is well-defined. By Lemma 
\eqref{key lemma}, with probability at least $1-l(4e^{-\gamma m_i}+1/{n^3})$, we have for $i=1,...,l$
\[
\|\mtx{X_i}\|_F\leq{1\over 5}\|\mtx{X_{i-1}}\|_F, ~~~\|{\mtx{Y_i}}_{T^\perp\cap \Omega}\|\leq {1\over 10}\|\mtx{X_i}\|
_F, ~~\text{and}~~\|{\mtx{Y_i}}_{\Omega^{\perp}}\|_\infty\leq {{C_0\sqrt{\log n}}\over{\sqrt{m}}}\|\mtx
{X_i}\|_F,
\]
provided $m_1\geq C_1k, ..., m_l\geq C_1k$. Therefore, $\mtx{Y}=v_1\vct{z_1}\vct{z_1}^T+...+v_m\vct{z_m}\vct{z_m}^T$  
and
\[
\|\mtx{Y}_{T\cap \Omega}-X_0\|_F=\|\mtx{X_l}\|_F\leq ({1 \over 5})^l\|\mtx{X_0}\|_F<{{\|\mtx{X_0}\|_F}\over{6n^2}},~~~
(\text{by $l>2\log n+2$})
\]
\[
\|\mtx{Y}_{T^\perp \cap \Omega}\|\leq \sum_{i=1}^l\|{\mtx{Y_i}}_{T^\perp \cap \Omega}\|\leq \sum_{i=1}^l{{\|\mtx{X_
{i-1}}\|_F}\over {10}}\leq\sum_{i=1}^l{{\|\mtx{X_0}\|_F}\over {10}}({1\over 5})^{(i-1)}\leq {{\|\mtx{X_0}\|_F}
\over 
8},
\]
and
\begin{align*}
\|Y_{\Omega^{\perp}}\|_\infty   &\leq \sum_{i=1}^l\|Y_{i_{\Omega^{\perp}}}\|\leq \sum_{i=1}^l {{C_0\sqrt{\log n}\|X_
{i-1}\|_F}\over{\sqrt{m}}}\leq {5\over 4}{{C_0\sqrt{\log n}}\over{\sqrt{m}}}\|\mtx{X_0}\|_F.
\end{align*}
When $m\geq (2\log n+3)C_1k$, we can always make such a division of $\{\vct{z_1}, ..., \vct{z_m}\}$, so the proof is 
complete.
\end{proof}


\section{The proof of Theorem \ref{thm2}}
We first prove a useful lemma:
\begin{lemma}
\label{lemma3.1}
Suppose $\vct{a_j},~j=1,...,m_1$ and $\vct{b_j}, ~j=1,..., m_2$ are IID $\mathcal{N}(0, \mtx{I_{N\times N}})$ random vectors in $
\mathbb{R}^N$, where $m_1\geq 0$, $m_2\geq 0$ and $m_1+m_2<N$. Then there is an event 
\[
\widetilde{E} =\widetilde{E}(\vct{a_1},..., \vct{a_{m_1}}, \vct{b_1},..., \vct{b_{m_2}})
\]
with probability at least $1-m_2e^{-0.09(N-m_1)}$, such that on $\widetilde{E}$ we have the following property:\\
Any $\alpha_j<0, ~ j=1,..., m_1$,  $\beta_j\geq 0, ~ j=1,..., m_2$, $\lambda \in \mathbb{R}$, $\mtx{S} \preceq \mtx{0}$ and $\mtx{L} \in \mathbb{R}^{n \times n}$ symmetric satisfying
\[
\sum_{j=1}^{m_1}\alpha_j\vct{a_j}\vct{a_j}^T+\sum_{j=1}^{m_2}\beta_j\vct{b_j}\vct{b_j}^T=\mtx{L}+\mtx{S}+\lambda\mtx{I},
\]
must also satisfy 
\[
{{N-m_1}\over 2}(\sum_{j=1}^{m_2}\beta_j)\leq \lambda m_2+\sqrt{m_2}\|\mtx{L}\|_F.
\]
\end{lemma}

\begin{proof}
With probability $1$ we have $\vct{a_1}, ..., \vct{a_{m_1}}, \vct{b_1}, ..., \vct{b_{m_2}}$ are linearly independent. Suppose 
$\{\vct{v_1}, ..., \vct{v_{m_1}}, \vct{v_{m_1+1}}, ..., \vct{v_{m_1+m_2}}, ..., \vct{v_N}\}$ is an orthonormal basis of 
$\mathbb{R}^N$ satisfying
\[
\Span{(\vct{a_1}, ..., \vct{a_{m_1}})}=\Span{(\vct{v_1}, ..., \vct{v_{m_1}})},
\]
and
\[
\Span{(\vct{a_1}, ..., \vct{a_{m_1}},~\vct{b_1}, ..., \vct{b_{m_2}})}=\Span(\vct{v_1}, ..., \vct{v_{m_1+m_2}}).
\]
Then we can further assume $(\vct{v_1}, ..., \vct{v_{m_1}})$ only depend on $(\vct{a_1}, ..., \vct{a_{m_1}})$ and 
are independent of $(\vct{b_1}, ..., \vct{b_{m_2}})$.
Then we have
\begin{align*}
\left\langle\sum_{j=m_1+1}^{m_1+m_2} \vct{v_j}\vct{v_j}^T, \mtx{L}+\mtx{S}+\lambda\mtx{I}\right\rangle
&=\left\langle\sum_{j=m_1+1}^{m_1+m_2} \vct{v_j}\vct{v_j}^T, \sum_{j=1}^{m_1}\alpha_j\vct{a_j}\vct{a_j}^T+\sum_{j=1}^
{m_2}\beta_j\vct{b_j}\vct{b_j}^T\right\rangle=\left\langle\sum_{j=m_1+1}^{m_1+m_2} \vct{v_j}\vct{v_j}^T, \sum_{j=1}^
{m_2}\beta_j\vct{b_j}\vct{b_j}^T\right\rangle\\
&=\left\langle\sum_{j=m_1+1}^N \vct{v_j}\vct{v_j}^T, \sum_{j=1}^{m_2}\beta_j\vct{b_j}\vct{b_j}^T\right\rangle=\left\langle\mtx{I}-\sum_
{j=1}^{m_1} \vct{v_j}\vct{v_j}^T, \sum_{j=1}^{m_2}\beta_j\vct{b_j}\vct{b_j}^T\right\rangle\\
&=\sum_{j=1}^{m_2}\beta_j\left(\|\vct{b_j}\|^2-\sum_{k=1}^{m_1}|\vct{v}_k^T\vct{b}_j|^2\right).
\end{align*}
Since $\vct{b_j}$ are IID $\mathcal{N}(\vct{0}, \mtx{I})$ random vectors, and are independent from the orthonormal vectors $\vct{v_1}, ..., \vct{v_{m_1}}$, we have
\[
\|\vct{b_j}\|^2-\sum_{k=1}^{m_1}|\vct{v}_k^T\vct{b}_j|^2 \sim \chi^2(N-m_1).
\]
By the Chernoff upper bound for the $\chi^2$ distribution, we have
\[
\P\left(\|\vct{b_j}\|^2-\sum_{k=1}^{m_1}|\vct{v}_k^T\vct{b}_j|^2 \geq {{N-m_1}\over 2}\right)\leq ({1\over 2}e^{1/2})^{(N-m_1)/2}\leq e^{-0.09(N-m_1)}.
\]
Then we have
\[
\left\langle\sum_{j=m_1+1}^{m_1+m_2} \vct{v_j}\vct{v_j}^T, \mtx{L}+\mtx{S}+\lambda\mtx{I}\right\rangle\geq \sum_{j=1}^{m_2} \beta_j({{N-m_1}\over 2})
\]
with probability $1-m_2e^{-0.09(N-m_1)}$.\\
\\
On the other hand, we have
\begin{align*}
\left\langle\sum_{j=m_1+1}^{m_1+m_2} \vct{v_j}\vct{v_j}^T, \mtx{L}+\mtx{S}+\lambda\mtx{I}\right\rangle  \leq \left\langle\sum_
{j=m_1+1}^{m_1+m_2} \vct{v_j}\vct{v_j}^T, \mtx{L}+\lambda\mtx{I}\right\rangle  \leq \lambda m_2+\|\mtx{L}\|_F \sqrt{m_2},
\end{align*}
which implies our claim.
\end{proof}

\paragraph{The proof of Theorem \ref{thm2}:}~~\\
We start by defining the event $E=E(\vct{z_1}, ..., \vct{z_m})$. First, we define an event
 \[
 E_0=\{\left\vert\left<\vct{x}, {\vct{z_j}}_G\right>\right\vert^2\leq 10\log n, j=1,..., m\}.
 \]
 By the assumption that $\|\vct{x}\|_2=1$ and $\vct{z_{j_G}} \sim \mathcal{N}(\vct{0}, \mtx{I}_{k \times k})$, we have
\[
\left\vert\left<\vct{x}, {\vct{z_j}}_G\right>\right\vert^2 \sim \chi^2(1),
\]
which implies that $\P(E_0)\geq 1-{m\over{n^5}}$. \\
Next, for any partition of $\{1,...,m\}=\{j_1,...,j_{m_1}\}\cup\{k_1,...,k_{m_2}\}$, where $j_1<...<j_{m_1}$, $k_1<...<k_{m_2}$, $m_1\geq 0$, $m_2\geq 0$ and $m_1+m_2=m$, define
\[
E_{\{j_1,...,j_{m_1}\}\cup\{k_1,...,k_{m_2}\}}=\widetilde{E}({\vct{z_{j_1}}}_B, ..., {\vct{z_{j_{m_1}}}}_B, {\vct{z_{k_1}}}_B, ..., {\vct{z_{k_{m_2}}}}_B).
\]
Then by Lemma \ref{lemma3.1} we have 
\[
\P(E_{\{j_1,...,j_{m_1}\}\cup\{k_1,...,k_{m_2}\}})\geq 1-m_2e^{-0.09(n-k-m_1)}\geq 1-me^{-0.09(n-k-m)}.
\]
Now we define the event $E$ by
\[
E=E_0\cap\left(\bigcap_{\text{all partitions of $[m]$}} E_{\{j_1,...,j_{m_1}\}\cup\{k_1,...,k_{m_2}\}}\right).
\]
Then
\[
\P(E)\geq 1-{m\over{n^5}}-2^m m e^{-0.09(n-k-m)}\geq 1-{m\over{n^5}}-m e^{-0.09n+0.09k+0.79m}.
\]
Hereafter all our discussions will be on the event $E$.\\
\\
We now come back to derive the necessary condition for $\vct{x}\vct{x}^T$ to be an optimal point 
of \eqref{eq:traceL1min}. By section 5.9.2 of \cite{Boyd}, the condition is
\[
\mtx{0} \in \partial(\|\mtx{X}\|_1 + \lambda \tr(\mtx{X}))|_{\vct{x}\vct{x}^T} + \mtx{S} + \cA^*(\vct{v}), \quad \mtx{S} \preceq \mtx{0}, 
\quad \left< \mtx{S},  \vct{x}\vct{x}^T\right> = 0 
\]
which, using the definition of the subgradient, is equivalent to
\[
\mtx{0} \in \sgn(\vct{x}\vct{x}^T) + \mtx{L}_{\Omega^\perp} + \lambda \mtx{I} + \mtx{S} + \cA^*(\vct{v}), \quad \mtx{S} \preceq \mtx
{0}, \quad \left< \mtx{S}, \vct{x}\vct{x}^T\right> = 0, \quad \|\mtx{L}_{\Omega^\perp}\|_\infty \leq 1 \\
\]
One can verify that $\mtx{S} \preceq \mtx{0}$ and  $\left< \mtx{S}, \vct{x}\vct{x}^T\right> = 0 $ is equivalent to $\mtx{S} \preceq \mtx{0}$ 
and $ \PT(\mtx{S}) = 0$. 
Thus the necessary condition for $\vct{x}\vct{x}^T$ to be a minimizer of this program is the existence of a dual certificate $\mtx{Y} $ with the following properties:
\begin{align}
\label{eqconstraint}
&\mtx{Y} = \sum_{j=1}^{m}c_j \vct{z_j }\vct{z_j}^T = \sgn(\vct{x})\sgn(\vct{x})^T + \mtx{L}_{\Op} + \lambda \mtx{I} + \mtx{S}_
{\Tp},\\
\label{ineqconstraint1}
&\|\mtx{L}_{\Op}\|_{\infty} \leq 1,\\
\label{ineqconstraint2}
&\mtx{S}_{\Tp} \preceq \mtx{0}.
\end{align}
Project both sides of \eqref{eqconstraint} on $\Gamma$, we have
\begin{equation}
\label{keyequation1}
\mtx{Y}_\Gamma=\sum_{j=1}^{m}c_j {\vct{z_j }}_B{\vct{z_j}}_B^T = \mtx{L}_{\Gamma} + \lambda \mtx{I}_\Gamma + \mtx{S}_\Gamma.
\end{equation}
Since $\Gamma \in T^\perp$, we have 
\begin{equation}
\label{S_Gamma}
\mtx{S}_\Gamma \preceq \mtx{0}.
\end{equation}
 It is also obvious that $\|\mtx{L}_{\Gamma} \|_\infty\leq \|\mtx{L}_{\Omega^\perp} \|_\infty\leq 1$, which implies
 \begin{equation}
 \label{L_Gamma}
 \|\mtx{L_\Gamma}\|_F\leq (n-k)\|\mtx{L_\Gamma}\|_\infty\leq n-k, \text{~and~} \tr(\mtx{L_\Gamma})\leq n-k.
 \end{equation}
 On the other hand, project both sides of \eqref{eqconstraint} on $T$, we have
\begin{align*}
\mtx{Y}_T =\|\vct{x}\|_1 (\sgn(\vct{x})\vct{x}^T + \vct{x}\sgn(\vct{x})^T) - \|\vct{x}\|_1^2 \vct{x}\vct{x}^T+  \mtx{L}_{T
\cap{\Omega^\perp}} +\lambda \vct{x}\vct{x}^T ,
\end{align*}
and
\begin{align*}
\mtx{Y}_{T\cap \Omega} =\|\vct{x}\|_1 (\sgn(\vct{x})\vct{x}^T + \vct{x}\sgn(\vct{x})^T) -\|\vct{x}\|_1^2\vct{x}\vct{x}^T  
+\lambda \vct{x}\vct{x}^T ,
\end{align*}
which implies
\begin{equation}
\label{keyequation2}
\vct{x}^T\mtx{Y}_{T\cap \Omega}\vct{x} = \sum_{j=1}^{m}c_j \left\vert\left<\vct{x}, {\vct{z_j}}_G\right>\right\vert^2  = \|
\vct{x}\|_1^2 +\lambda \|\vct{x}\|_2^2= \|\vct{x}\|_1^2 + \lambda.
\end{equation}

 \paragraph{Case 1: $\lambda<-{k\over2}$.}~~ \\
By the assumption $k\leq m \leq {n \over {40\log n}}$, we can assume the eigenvalue decomposition
\[
\sum_{j=1}^{m}c_j {\vct{z_j }}_B{\vct{z_j}}_B^T =\mu_1\vct{u_1}\vct{u_1}^T+...+\mu_m\vct{u_m}\vct{u_m}^T+0\cdot\vct{u_
{m+1}}\vct{u_{m+1}}^T+...+0\cdot \vct{u_{n-k}}\vct{u_{n-k}}^T,
\]
where $\{\vct{u_1}, ..., \vct{u_{n-k}}\}$ is an orthogonal basis of $\Span(\vct{e_{k+1}}, ..., \vct{e_{n}})$.
Then by \eqref{keyequation1}, we have
\[
\vct{u_j}^T(\mtx{L}_{\Gamma} + \lambda \mtx{I}_\Gamma + \mtx{S}_\Gamma)\vct{u_j}=\vct{u_j}^T\left(\sum_{j=1}^{m}c_j {\vct{z_j }}_B{\vct{z_j}}_B^T \right)\vct{u_j}=0,
\]
for $j=m+1,.., n-k$. By \eqref{S_Gamma}, we have
\begin{equation}
\label{bigterm}
\vct{u_j}^T\mtx{L_{\Gamma}}\vct{u_j}\geq {{\|\vct{x}\|_1^2}\over2}=-\vct{u_j}^T(\lambda \mtx{I}_\Gamma + \mtx{S}_\Gamma)\vct{u_j}\geq -\lambda\geq {k \over 2} 
\end{equation}
Since $\{\vct{u_1}, ..., \vct{u_{n-k}}\}$ is an orthogonal basis of $\Span(\vct{e_{k+1}}, ..., \vct{e_{n}})$, we have
\[
\sum_{j=1}^{n-k}\vct{u_j}^T\mtx{L_{\Gamma}}\vct{u_j}=\left\langle \mtx{L_\Gamma}, \sum_{j=1}^{n-k}\vct{u_j}\vct{u_j}^T\right\rangle=\tr(\mtx{L_\Gamma})\leq n-k.
\]
By \eqref{bigterm} and the assumption $4\leq k\leq m \leq {n \over {40\log n}}$, we have 
\begin{equation}
\label{smallterm}
\sum_{j=1}^{m}\vct{u_j}^T\mtx{L_{\Gamma}}\vct{u_j}\vct{u_m}\leq n-k-(n-k-m){k \over2}<0.
\end{equation}
On the other hand
\begin{equation}
\label{Fcontrol}
\left|\sum_{j=1}^{m}\vct{u_j}^T\mtx{L_{\Gamma}}\vct{u_j}\vct{u_m}\right|=\left|\left\langle \mtx{L_\Gamma}, \sum_{j=1}^{m}\vct{u_j}\vct{u_j}^T\right\rangle\right|\leq \|\mtx{L_\Gamma}\|_F\|\sum_{j=1}^{m}\vct{u_j}\vct{u_j}^T\|_F\leq (n-k)\sqrt{m}.
\end{equation}
By \eqref{smallterm} and \eqref{Fcontrol}, we have $\leq (n-k)\sqrt{m}\geq (n-k-m){k \over2}-(n-k)$ which implies
\[
m\geq ({k\over 4}-1)^2.
\]

\paragraph{Case 2:  $\lambda \geq -{k\over2}$.}~~\\ 
Let $I_+ = \{k \in \{1,2\ldots,m\}; c_k \geq 0\} $ and $I_- = \{k \in 
\{1,2,\ldots,m \}; c_k < 0\} $. By \eqref{keyequation2} and the definition of $E \subset E_0$, we have
\begin{equation}
\label{low}
\|\vct{x}\|_1^2+\lambda\leq 10 \text{log}(n) \sum_{j \in I_+}c_j .
\end{equation}
By \eqref{keyequation2},
\[
\mtx{Y}_\Gamma=\sum_{j \in I_-}c_j {\vct{z_j }}_B{\vct{z_j}}_B^T +\sum_{j \in I_+}c_j {\vct{z_j }}_B{\vct{z_j}}_B^T = \mtx{L}_\Gamma + \lambda \mtx{I}_\Gamma + \mtx{S}_\Gamma.
\]
By the definition of $E$ and Lemma \ref{lemma3.1}, we have
\begin{equation}
\label{up}
{{n-k-|I_-|}\over 2}\sum_{j \in I_+}c_j \leq \lambda |I_+|+\sqrt{|I_+|}\|\mtx{L_\Gamma}\|_F.
\end{equation}
Notice that $\|\mtx{L_\Gamma}\|_F\leq (n-k)\|\mtx{L_\Gamma}\|_\infty\leq n-k$. By \eqref{low} and \eqref{up},
\[
{{(n-k-m)\|\vct{x}\|_1^2}\over{20\log n}}+\lambda\left({{n-k-m}\over{20\log n}}-m\right)\leq \sqrt{m}(n-k).
\]
By the assumption that $k\leq m \leq {n\over{40\log n}}$ and $\lambda \geq -{k\over2}$, we have 
\[
{{(n-k-m)(\|\vct{x}\|_1^2-k/2)}\over{20\log n}}\leq \sqrt{m}(n-k),
\]
which implies
\[
m\geq  {{\max(\|\vct{x}\|_1^2-k/2, 0)^2}\over{500\log^2 n}}.
\]
\\
Therefore, by putting Case 1 and Case 2 together, we have
\[
m\geq  \min\left(({k\over 4}-1)^2, {{\max(\|\vct{x}\|_1^2-k/2, 0)^2}\over{500\log^2 n}}\right).
\]

\section{Discussion}
We provide theoretical guarantees on the recovery of a sparse signal from quadratic Gaussian measurements via convex programming and show that our results are sharp for a class of recently proposed convex relaxations. For this model, unlike classical compressed sensing, compressive phase retrieval imposes a stricter limitation on the number of measurements needed for recovery via naive convex relaxation than is needed for well-posedness. This leads to a natural open question: can we narrow the gap by using other convex programs besides \eqref{eq:traceL1min}?\\

 
Theorem \ref{thm2} shows the limitations of \eqref
{eq:traceL1min} in the sense of exact recovery, since we only need to recover the support of the unknown 
vector to recover $\vct{x}$  by using the PhaseLift algorithm \cite{PhaseLift2, ImprovedPhaseLift} to solve the resulting  
overdetermined system of quadratic equations. Mathematically, recovering the support is at least as easy as 
exact recovery. Can we do better  than \eqref{eq:traceL1min} by formulating the right support recovery problem? We leave these 
considerations for future research. \\

\small
\section{Acknowledgements}
We are thankful for fruitful discussion with Emmanuel Cand\`{e}s and also Mahdi Soltanolkotabi, who generously provided us with the results of his numerical experiments on sparse recovery.

\bibliographystyle{plain} 
\bibliography{ref}

\end{document}